\documentclass[a4paper,5p,authoryear]{elsarticle}
\usepackage{amsmath, amsthm, amsfonts, amssymb, color}
\usepackage{mathrsfs}
\usepackage[UKenglish]{babel}
\theoremstyle{plain}
\newtheorem{theorem}{Theorem}[section]
\newtheorem{corollary}[theorem]{Corollary}
\newtheorem{lemma}[theorem]{Lemma}

\theoremstyle{definition}

\newtheorem{example}[theorem]{Example}
\newtheorem{definition}[theorem]{Definition}

\newcommand\nat{\mathbb N}

\newcommand\real{{\mathbb{R}}}
\newcommand\rd{{\mathbb{R}^d}}

\newcommand{\omu}[3]{\overset{\textrm{#1}}{\underset{\textrm{#3}}{#2}}}

\newcommand{\entier}{\mathop{\mathrm{int}}}
\newcommand\Ascr{\mathscr{A}}
\newcommand\Bscr{\mathscr{B}}

\newcommand\Fscr{\mathscr{F}}
\newcommand\Iscr{\mathscr{J}}

\newcommand\Pscr{\mathscr{P}}

\begin{document}\allowdisplaybreaks
\title{{\bfseries Continuity Assumptions in Cake-Cutting}}

\author[tud]{Ren\'e L.\ Schilling\corref{cor1}}
\ead{rene.schilling@tu-dresden.de}
\address[tud]{Institut f\"{u}r Mathematische Stochastik, Fachrichtung Mathematik,\\TU Dresden, 01062 Dresden, Germany}
\cortext[cor1]{Corresponding author}

\author[tuba]{Dietrich Stoyan}
\ead{stoyan@math.tu-freiberg.de}
\address[tuba]{Institut f\"{u}r Stochastik, Fakult\"{a}t f\"{u}r Mathematik und Informatik,\\TU Berg\-akademie Freiberg, Pr\"{u}ferstra{\ss}e 9, 09596 Freiberg, Germany}
%\thanks{}

%\subjclass[2010]{}

%\date{\today}

\begin{abstract}
In important papers on cake-cutting -- one of the key areas in fair division and resource allocation -- the measure-theoretical fundamentals are not
fully correctly given. It is not clear (i) which family of sets should be taken for the pieces of cake, (ii) which set-functions should be used for evaluating the pieces, and (iii) which is the relationship between various continuity properties appearing in cake-cutting.

We show that probably the best choice for the familiy
of subsets of $[0,1]$ is the Borel $\sigma$-algebra and
for the set-function any `sliceable' Borel measure. At least in dimension one it does not make sense
to work with only finitely additive contents on finite unions of intervals. For the continuity
property we see two possibilities. The weaker is the traditional divisibility property, which is equivalent to being atom-free.
The stronger is simply absolute continuity with respect to Lebesgue measure. We also consider the case of a base set
(cake or pie) more general than $[0,1]$.
\end{abstract}
\begin{keyword}
    Cake-cutting \sep
    finitely additive measure \sep
    %content \sep
    countably additive measure \sep
    non-atomic \sep
    sliceable \sep
    measurability

    \emph{JEL classification:} C7 %\sep D7 %C0 \sep C6 \sep D7

    \MSC[2010] 28A12 %\sep 91B08
    \sep 91B32
\end{keyword}

\maketitle

\section{Introduction}\label{intro}
Cake-cutting `is a metaphor for dividing an infinitely divisible resource (or, good)
amongst several players (or agents). As everyone knows, people's preferences do not
only differ with respect to taste but to all kind of things. Hence, the problem of
fair division of a heterogeneous divisible resource also applies to areas such as
economics \& law (e.g., the division of a property in a divorce or inheritance case),
science \& technology (e.g., the division of computing time among several users
sharing a single computer, or the division of bandwidth when sharing a network),
and even politics' \cite[p.~395]{rothe}.

Cake-cutting is a key area in fair division and resource allocation which offers many
deep mathematical problems with admirable solutions see, e.g., \cite{steinhaus48},
\cite{bra-tay96} and \cite{rothe}.

In contrast to the rigorous treatment of cake-cutting problems by
high-level mathematics, a subtle point at the very beginning of
problem description is in many papers considered with little care.
Two passages from recent papers may show this.

\cite{aziz} write in their very deep and important paper: `We consider a cake which is represented by the interval $[0,1]$. A \emph{piece of cake} is a
finite union of disjoint subsets of $[0,1]$. We will assume the standard assumptions in cake cutting. Each agent in the set of agents $N=\{1, \dots, n\}$ has his own valuation over subsets of interval $[0,1]$. The valuations are (i) \emph{defined on all finite unions of the intervals}; (ii) non-negative: $V_i(X) \ge 0$ for all $X \subseteq [0,1]$; (iii) \emph{additive}: for all disjoint $X, X' \subseteq [0,1], V_i(X \cup X')=V_i(X) + V_i(X')$; (iv) \emph{divisible} i.e., for every $X\subseteq [0,1]$ and $0 \le \lambda \le 1$, there exists $X' \subseteq X$ with $V_i(X') = \lambda V_i(X)$.'

The term \emph{finite union of disjoint subsets} is simply a
set-theoretic blunder. The term \emph{finite union of intervals}
causes questions to be discussed when combined with valuations,
which are set-functions with properties as measures or contents.
Such functions (unless trivial) cannot be defined on the system `of
all subsets' of a given uncountable set, not even for $[0,1]$.

\cite{bra-jon-kla11} write: `A cake is a one-dimensional
heterogeneous good, represented by the unit interval $[0,1]$. Each
of $n$ players has a personal value function over the cake,
characterised by a probability density function with a continuous
cumulative distribution function. This implies that players'
preferences are finitely additive and nonatomic.'

This text is closer to measure theory. However, a probability
distribution with a density function (an absolutely continuous
measure) has of course a (cumulative) distribution function\footnote{Throughout we use the terms \emph{distribution function} and \emph{cumulative distribution function} synonymously.} which is
continuous -- but the continuity of a distribution function does not
imply the existence of a density function.

We emphasise that these inaccuracies do, by no means, influence the results in these papers.\footnote{Note however, that imprecise language may lead to serious mistakes as shown in \cite{hil-mor09,hil-mor10}.} We only aim to discuss this small but subtle blemish, and to show how the measure-theoretic fundamentals of cake-cutting should be correctly formulated.

When comparing the two quotations above, three questions arise:
\begin{enumerate}
\item
    Which subsets of $[0,1]$ should be considered? Finite unions of intervals or more general sets?
\item
    If valuations are considered as set-functions studied in measure theory, should they be countably additive measures or finitely additive contents?
\item
    Which relationship is bet\-ween the two continuity assumptions above, between divisibility and the existence of a density function?
\end{enumerate}

We come to clear answers on these questions and a recommendation how
to formulate the fundamentals of cake-cutting.

\section{Measure theory in a nutshell}\label{meas}

This section introduces basic concepts of (abstract) measure theory. For an easily accessible full account we refer to \cite{bauer}, \cite{cohn} or \cite{schilling-mims,schilling-mint}. Let $X$ be any abstract set and consider a function $\mu:\Ascr\to [0,\infty]$ defined on a family $\Ascr$ of subsets of $X$. Motivated by the notions `length', `area', `volume', or `number of elements' in a set etc., we call $\mu$ \emph{finitely additive} if
\begin{equation}
    \mu(A\cup B) = \mu(A)+\mu(B)
\end{equation}
holds for all disjoint $A,B\in\Ascr$ such that $A\cup B\in \Ascr$.
This easily extends to finitely many mutually disjoint sets.
%$A_1,\dots, A_n$.
For countably many sets we have to require the
so-called \emph{$\sigma$-additivity}
\begin{equation}
    \mu\Big(\bigcup_{n=1}^\infty A_n\Big) = \sum_{n=1}^\infty \mu(A_n)
\end{equation}
for all sequences of mutually disjoint sets $(A_n)_{n\in\nat}\subset\Ascr$ such that $\bigcup_{n=1}^\infty A_n\in \Ascr$.

If $\emptyset\in\Ascr$ and $\mu(\emptyset)<\infty$, then finite additivity (or $\sigma$-additivity) implies that
$\mu(\emptyset)=0$.
%Indeed, $$ 0\leq \mu(\emptyset)=\mu(\emptyset\cup\emptyset)=\mu(\emptyset)+\mu(\emptyset)<\infty. $$
It is quite cumbersome to require in each instance the stability of $\Ascr$ under finitely many or countably many unions. Therefore, it is usual to consider set-functions on algebras and $\sigma$-algebras of sets. An \emph{algebra} $\Ascr$ of subsets of $X$ is a family satisfying
\begin{equation}
\begin{gathered}
    \emptyset\in \Ascr,\quad
    A\in\Ascr \implies A^c = X\setminus A\in \Ascr,\\
    A,B\in\Ascr \implies A\cup B\in\Ascr,
\end{gathered}
\end{equation}
while a $\sigma$-algebra satisfies
\begin{equation}
\begin{gathered}
    \emptyset\in \Ascr,\quad
    A\in\Ascr \implies A^c = X\setminus A\in \Ascr,\\
    (A_n)_{n\in\nat}\subset \Ascr \implies \bigcup_{n=1}^\infty A_n\in\Ascr.
\end{gathered}
\end{equation}
These definitions allow to perform any of the usual operations with sets -- union,
intersection, set-theoretic difference, formation of complements -- finitely often
(in the case of an algebra) or countably often (in the case of a $\sigma$-algebra).

\begin{definition}
(i) A \emph{content} is a finitely additive set-function on an algebra such that $\mu(\emptyset)=0$.

(ii) A \emph{measure} is a $\sigma$-additive set-function on a $\sigma$-algebra such that $\mu(\emptyset)=0$.
\end{definition}
Contents and measures are \emph{subtractive}, i.e.\ if $A\subset B$, then $\mu(B\setminus A)=\mu(B)-\mu(A)$; in particular they are monotone, i.e.\ $\mu(A)\leq\mu(B)$ whenever $A\subset B$; $\sigma$-additivity is equivalent to finite additivity plus \emph{continuity from below}, i.e.\ if $B_1\subset B_2\subset \dots$ is a chain of countably many sets from $\Ascr$, then we have $\mu\big(\bigcup_{n=1}^\infty B_n\big) = \sup_{n\in\nat}\mu(B_n)$.

\begin{example}\label{example-1}
Here are a few typical examples of measures.

(i) \emph{$d$-dimensional Lebesgue measure} is defined on the Borel $\sigma$-algebra on $\rd$; this is the smallest $\sigma$-algebra which contains the $d$-dimensional `rectangles' of the form
$$
    (a_1,b_1] \times \dots \times (a_d,b_d],\quad -\infty < a_i<b_i<\infty.
$$ It is the unique measure which corresponds to the familiar geometric notions of length (if $d=1$), area (if $d=2$), volume (if $d=3$) etc.

(ii) The \emph{Dirac measure} (or \emph{point mass}) $\delta_{x_0}$
can be defined on the power set $\Pscr(X) = \{A : A\subset X\}$ of
any abstract space $X$. It is an indicator as to whether a fixed
point $x_0\in X$ is in a set $A$: $\delta_{x_0}(A) = 1$ or $0$
depending on whether $x_0\in A$ or $x_0\notin A$.

(iii) The \emph{counting measure} can be defined on the power set $\Pscr(X) = \{A : A\subset X\}$ of any abstract space $X$. Its value
is the number of elements of $A\subset X$: $\mu(A) = \# A$. Formally, it can be written as a (perhaps uncountable) sum $\sum_{x\in X}\delta_x$.
\end{example}

Despite its innocuous definition, a $\sigma$-algebra can become quite large and complicated if the base set $X$ is non-discrete. This means that we have, in general, no possibility to assign explicitly a measure to \emph{every} set from a $\sigma$-algebra. One way out is Carath\'eodory's `automatic extension theorem' \cite[Theorem 6.1]{schilling-mims}.

\begin{theorem}\label{th-cara}
    Assume that $\mu$ is a finite \textup{(}i.e.\ $\mu(X)<\infty$\textup{)} content on an algebra $\Ascr$. Then there is a unique measure $\bar\mu$ which is defined on the smallest $\sigma$-algebra containing $\Ascr$ and which coincides with $\mu$ on $\Ascr$.
\end{theorem}
With a bit more effort, one can improve this extension theorem to finitely-additive set-func\-tions $\mu$ which are defined on \emph{semi-rings of sets}; a semi-ring $\Iscr$ is a family of subsets of $X$ such that $\emptyset\in\Iscr$, $R,S\in\Iscr\implies R\cap S\in\Iscr$ and $R\setminus S$ can we written as finite union of mutually disjoint elements of $\Iscr$. Archetypes of semi-rings are the families of half-open intervals $(a,b]$, $a<b$, or $d$-dimensional rectangles.

Every content which is initially defined on a semi-ring $\Iscr$ can be (constructively!) extended to the smallest algebra $\Ascr$ containing $\Iscr$, cf.\ \cite[p.\ 39, Step 2]{schilling-mims}, and then we are back in the setting of the theorem. This allows us to give, in dimension one, important examples of contents and measures, cf.\ Lemma~\ref{lemma-1} below.

\bigskip
There is a well-known trade-off between the complexity of the measure $\mu$ and the size of the $\sigma$-algebra $\Ascr$. This is best illustrated by an example. While we can define the Dirac and counting measures for \emph{all} sets, there is no way to define a geometrically sensible notion of `volume' (i.e.\ Lebesgue measure) for all sets -- if we accept the validity of the axiom of choice. (One can even show that the axiom of choice is equivalent to the existence of non-measurable sets, cf.\ \cite[p.~55]{cie89}.)

In dimensions $d=3$ and higher this is most impressively illustrated by the following
\begin{example}[Banach--Tarski paradox]\label{example-2}
    Think of the closed unit ball $B(0,1) = \{x\in\real^3 : |x| \leq 1\}$ as a three-dimensional pie. Then there exists a decomposition of $B(0,1)$ into finitely many pieces $E_1\cup\dots\cup E_n$ such that there are geometrically congruent pieces $F_1, \dots, F_n$ which can be re-assembled into a pie $B(0,2)$ of twice the original radius! Since congruent sets should have the same volume (Lebesgue measure), this construction is only possible if the sets $E_i$, hence the sets $F_i$, are non-measurable. For full details we refer to \cite{wagon}.
\end{example}

\section{One-dimensional cakes}\label{cakes}

If $X=\real$ or $X=[0,1]$, one usually takes the Borel $\sigma$-algebra, which is the smallest $\sigma$-algebra containing all open (or half-open) intervals. A content is often defined through its values on the half-open intervals
$$
    \Iscr = \{(a,b] : -\infty \leq a < b<\infty\}\cup\{\emptyset, \real\},
$$
cf.\ the remark following Theorem~\ref{th-cara}.

Let $\mu$ be a measure on the Borel $\sigma$-algebra $\Ascr$ in $X\subset\real$. A \emph{null set} for the measure $\mu$ is a set $N\in\Ascr$ such that $\mu(N)=0$. Recall that a finite measure $\mu$ is called \emph{absolutely continuous} with respect to Lebesgue measure, if every Lebesgue null set is also a $\mu$ null set; this is equivalent to the existence of an integrable density function $f\geq 0$ such that $\mu(B) = \int_B f(x)\,\mathrm{d}x$ for all $B\in\Ascr$; $\mu$ is said to be \emph{continuous} if its %right-continuous
\emph{distribution function} $F_\mu(x):=\mu(-\infty, x]$ is continuous; we use $\mu((-\infty,x]\cap I)$ if $\mu$ lives on an interval $I\subset\real$.
%Let us point out that $\sigma$-additivity of a measure $\mu$ implies that the distribution function $F_\mu$ of the measure $\mu$ is always continuous from the right.
A typical example of a continuous but not absolutely continuous measure is the measure which is induced by \emph{Cantor's function} (aka the devil's staircase), cf.\ Figure~1. Since $\mu\{x\} = F(x)-F(x-)=F(x)-\lim_{t\uparrow x}F(t)$, the notions of `continuity' and `atom-free' coincide in the one-dimensional setting.
\begin{figure}[h!]
    \includegraphics[width = .95\linewidth]{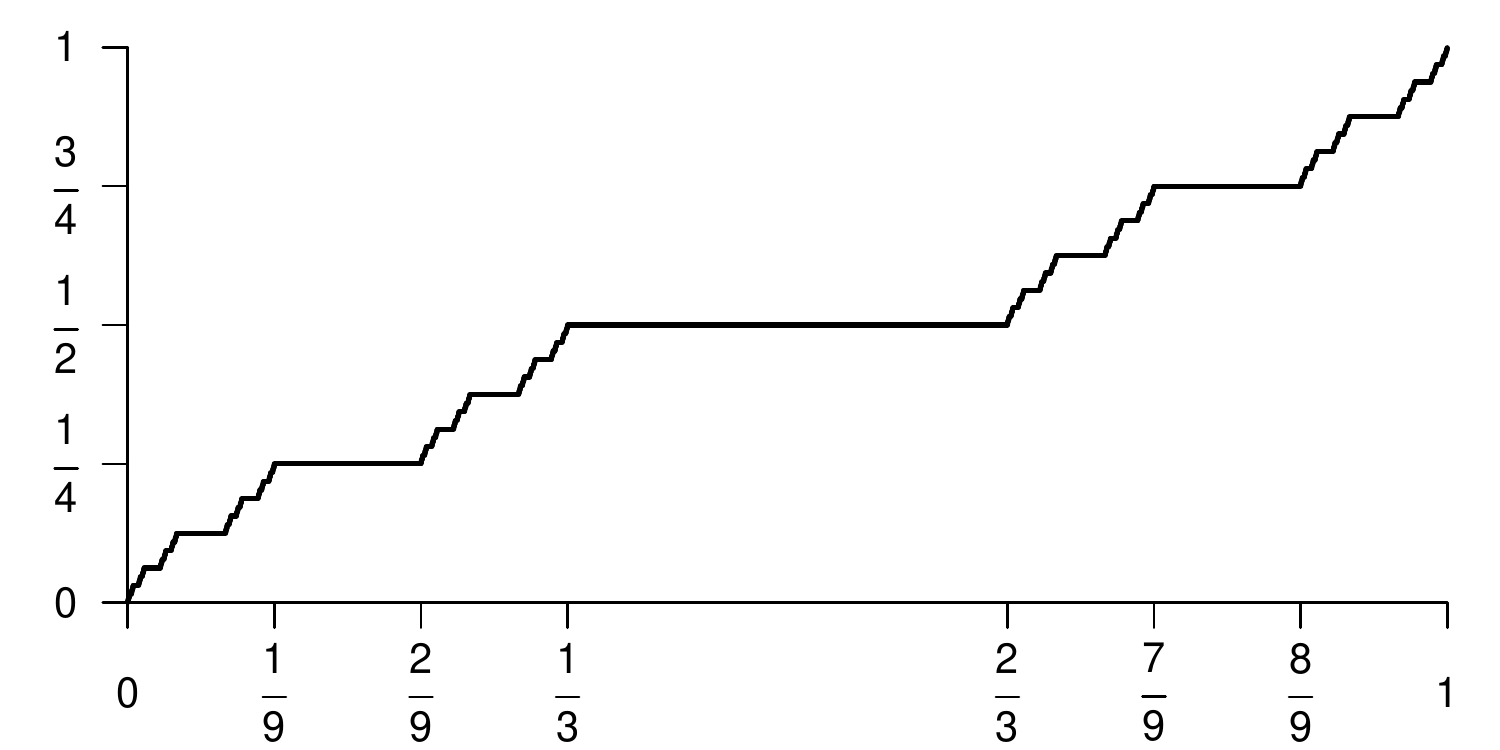}
    \caption{Let $C$ be Cantor's ternary set (i.e.\ the infinite triadic fractions $0.x_1x_2x_3\dots$ where $x_i\in\{0,2\}$. In order to enforce uniqueness, we identify $0.1$ with $0.0222\dots$ etc.) which is constructed by removing the open middle third from $[0,1]$ and then applying the same procedure to the remaining (disjoint) closed intervals \emph{ad infinitum}. The Cantor function takes the value $1/2$ on the initially removed middle third, $1/4$ and $3/4$ on the sets removed in the second step, and $i/2^n$, $i=1,3,5,\dots 2^n-1$ on the sets removed in the $n$th step. More formally, $f(x) = \sup\{\sum_{n=1}^\infty (x_n/2) 2^{-n} : x \geq 0.x_1x_2x_3\dots \in C\}$; the Cantor function is continuous, it increases exactly on $C$ and it is flat on $[0,1]\setminus C$.}
\end{figure}

The following lemma is well-known in measure theory, see e.g.\ \cite{schilling-mims,schilling-mint}, its second part describes the
relation between right-continuous distribution functions and probability measures. The first part (i) appears just as a part of the standard proof and is usually not explicitly spelled out; our point is to highlight the role of the right-continuity of the distribution function. Finally, (iii) means that a finitely-additive content with a (right-)continuous distribution function is already $\sigma$-additive.
\begin{lemma}\label{lemma-1}
    \textup{(i)} Let $\mu$ be a content defined on the system of all
    half-open intervals $\Iscr$ such that $\mu(\real) = 1$. The distribution function $F_\mu(x) := \mu(-\infty,x]$ is a monotonically increasing function such that $F_\mu(-\infty)=0$ and $F_\mu(+\infty)=1$.

    Conversely, every monotonically increasing function $F:\real\to [0,1]$ such that $F(-\infty)=0$ and $F(+\infty)=1$ defines a content on $\Iscr$ via $\mu(a,b] = F(b)-F(a)$ and $\mu(\real):=1$.

    \textup{(ii)} If $\mu$ is a measure, $F_\mu$ is, in addition, right-con\-tinu\-ous, i.e.\ $F_\mu(x) = F_\mu(x+) = \lim_{h\downarrow 0} F_\mu(x+h)$.

    \textup{(iii)} The content induced by a right-continuous, monotone increasing $F$ has a unique extension to a $\sigma$-additive measure.
\end{lemma}

Lemma~\ref{lemma-1} shows, in particular, that every content $\mu$ such that $\mu(a,b]$ is a continuous function of the end-points of the interval $(a,b]$ can be uniquely extended to a $\sigma$-additive measure. This continuity assumption is frequently made in connection with cake-cutting, cf.\ \cite[p.\ 497]{barbanel-et-al}.

Let us now turn to cake-cutting. The following divisibility assumption is quite standard, cf.\ \cite[S.\ 398]{rothe}.
\begin{definition}\label{def-D}
    A measure or a content $\mu$ defined on $\Ascr$ has the property \textup{(D)}
    if for every $A\in\Ascr$ and $\alpha\in(0,1)$ there is a set $L_\alpha = L_\alpha(A)\in\Ascr$
    such that $L_\alpha\subset A$ and $\mu(L_\alpha) = \alpha\mu(A)$.
\end{definition}

Assume that $X=\real$ or $X=[0,1]$, and $\Ascr=\Bscr(X)$ is the system of Borel sets. It is clear that condition \textup{(D)} immediately entails that $\mu$ has no atoms, i.e.\ $\mu\{x\}=0$ for all $x\in\real$; indeed, the set $A=\{x\}$ admits only $\{x\}$ and $\emptyset$ as subsets, so $\mu\{x\}=0$ since $\mu$ is a finite set-function. The following lemma shows that the converse holds for a measure, whereas for a content this need not be the case, cf.\  Example~\ref{example-3} below.

\begin{lemma}\label{lemma-2}
    Every continuous \textup{(}but not necessarily absolutely continuous\textup{)} measure $\mu$ on the real line $(\real,\Bscr(\real))$ enjoys the property \textup{(D)}. In particular, \textup{(D)} is equivalent to $\mu$ having no atoms, i.e.\ the distribution function $F_\mu(x) = \mu(-\infty,x]$ is continuous.
\end{lemma}
\begin{proof}%[Proof of Lemma~\ref{lemma-2}]
    Fix a Borel set $B$ with $\mu(B)>0$ and set $F_B(x):=\mu(B\cap (-\infty,x])/\mu(B)$. Since $\mu$ has no atoms, $F_B(x)$ is continuous with range $[0,1]$. Given $\alpha\in (0,1)$ there is at least one point $x_\alpha\in\real$ such that $F(x_\alpha)=\alpha$. Clearly, $B_\alpha := B\cap (-\infty,x_\alpha]$ is a Borel subset of $B$ and $\mu(B_\alpha)=\alpha\mu(B)$.

    Since \textup{(D)} guarantees that $\mu\{x\}=0$ for all $x\in\real$, we see that \textup{(D)} is equivalent to $\mu$ having no atoms.
\end{proof}

Lemma~\ref{lemma-2} is wrong for a content $\mu$ which is defined on the system of all half-open intervals $\Iscr$ and with the property that the endpoints of the intervals have zero content. Note that this allows us to extend $\mu$ consistently to any interval $(a,b)$, $(a,b]$, $[a,b)$ or $[a,b]$ as they will have the same content.
\begin{example}\label{example-3}
    If $\mu$ is a finite content defined on $\Iscr$ with the additional property that it has no atoms, i.e.\ $\mu\{x\}=0$ for all $x\in\real$, then $\mu$ need not have the property \textup{(D)}. Take, for simplicity $X=[0,1]$ and consider the contents $\mu_F$ and $\mu_G$ induced by the following distribution functions
    \begin{align*}
        F(x)
        &=
        \begin{cases}
        \lambda x, &0\leq x< \frac 12\\%[\bigskipamount]
        \frac 12, & x=\frac 12\\%[\bigskipamount]
        \lambda x + 1-\lambda, & \frac 12<x\leq 1,
        \end{cases}
    \intertext{or}
        G(x)
        &=
        \begin{cases}
        0, &x=0\\%[\bigskipamount]
        (1-2\epsilon)x+\epsilon, & 0<x<1\\%[\bigskipamount]
        1, & x= 1.
        \end{cases}
    \end{align*}
    Since $F$ (and $G$) is not everywhere right-continuous, we \emph{cannot} conclude that for $\big(0,\frac 12\big) = \bigcup_{n=1}^\infty \big(0,\frac 12-\frac1n\big]$ the first equality (marked with `?') in the following identity holds:
    $$
        \mu_F(0,\tfrac 12) \omu{?}{=}{} \lim_n \mu_F\big(0,\tfrac 12 - \tfrac1n\big] = \lim_n F\big(\tfrac 12-\tfrac 1n\big)=F(\tfrac 12-).
    $$
    Thus, there is no intrinsic way to extend $\mu_F$ to all open intervals. On the other hand, setting $\mu_F\{x\}:= 0$ for all $x$ gives a content which is defined on all intervals, without leading to a contradiction.

    Let us extend $\mu_F$ and $\mu_G$ to all open, closed and half-open intervals by setting $\mu_F\{x\}=\mu_G\{x\}:= 0$.

\begin{figure}[h!]
    \includegraphics[width = .45\textwidth]{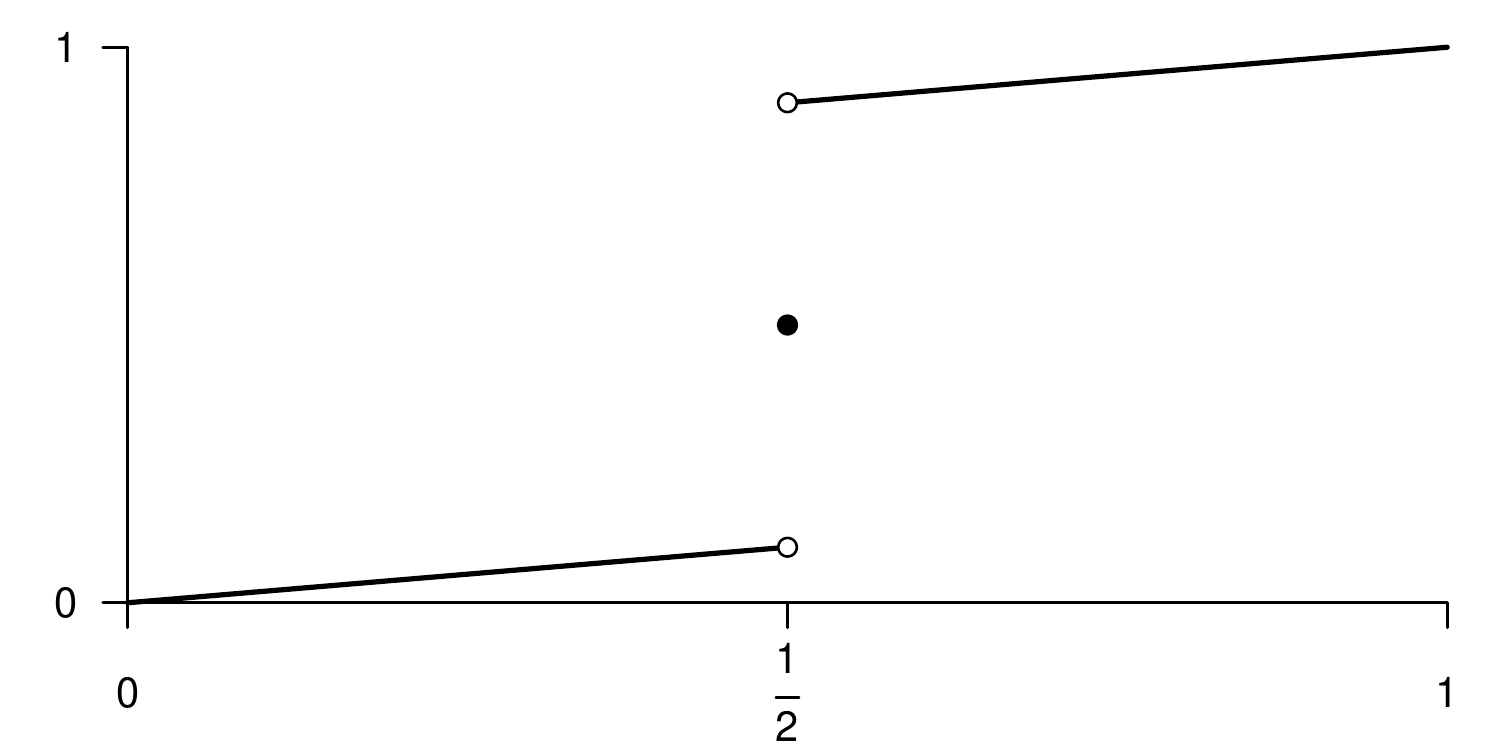}
    \hfill
    \includegraphics[width = .45\textwidth]{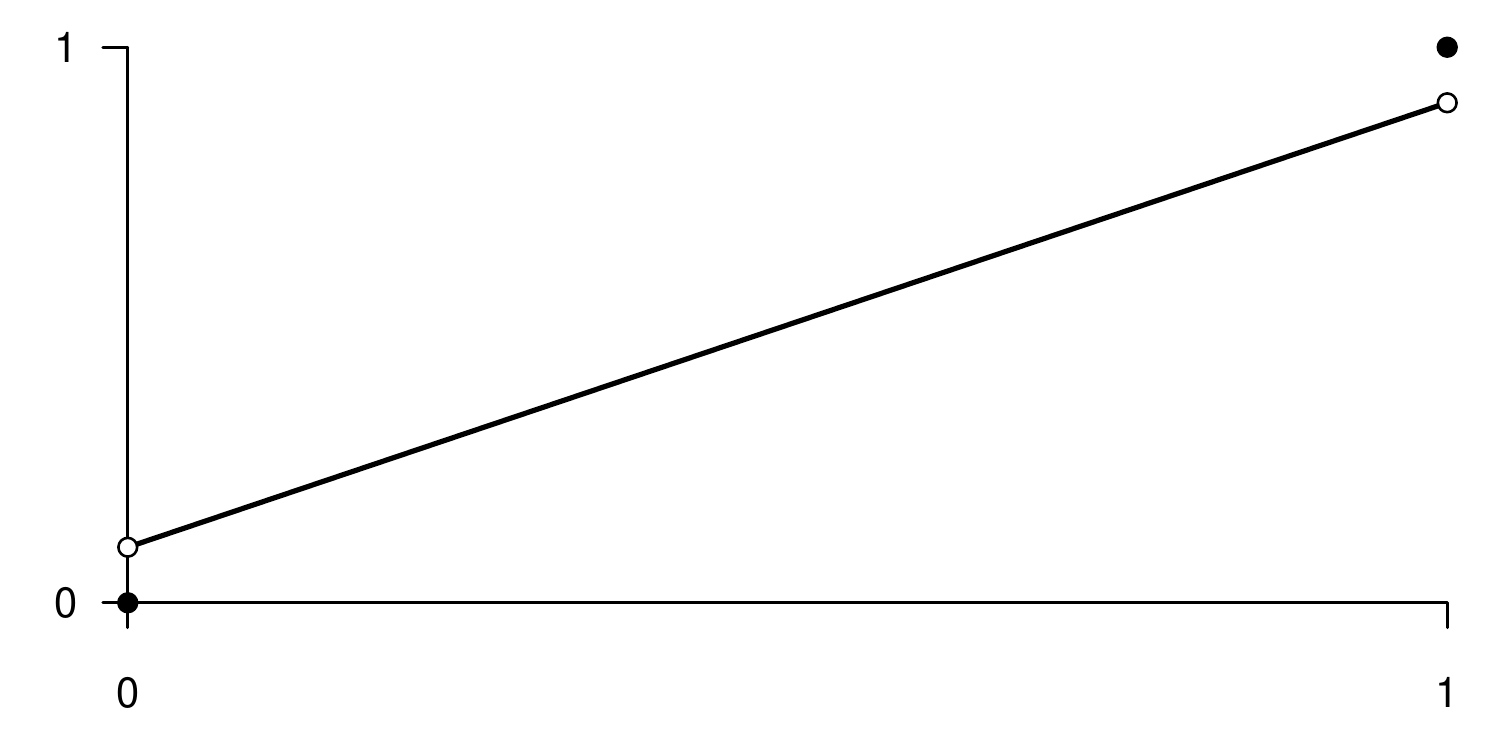}
\caption{The (cumulative) distribution functions $F$ (above) and $G$ (below).}
\end{figure}

    Take, in the first case, $\lambda= 0.1$;  then it is impossible to find a subset of $[0,1]$ of measure $0.2$. In the second case we take $\epsilon=0.05$. Then we can find sets of any size $\alpha < 1$, but, in general, they have to consist of at least two disjoint intervals, e.g.\ if $\alpha = 0.99$. However, both distribution functions can be used to define contents on the half-open intervals of $[0,1]$ such that all single points $\{x\}$ have zero content.
\end{example}

If one wants to include, in a sensible way, contents into
cake-cutting, one should replace \textup{(D)} by the assumption of continuity
of the (cumulative) distribution function of the underlying measure.
Then, alas, we are already in the situation of $\sigma$-additive
measures, see Lemma~\ref{lemma-1}, i.e.\ there is no real gain in
considering contents in dimension one.

\section{Abstract cakes and pies}\label{pies}

We are now going to consider abstract cakes, including `pies', i.e.\ $d=2$-di\-men\-sion\-al objects. Let us consider a general set $X$ (a `pie') equipped with an algebra or a $\sigma$-algebra $\Ascr$ of subsets of $X$. In order to assign a size to the pieces of pie $A\in\Ascr$, we use a content or a measure $\mu$ defined on $\Ascr$. If $\mu(X)<\infty$, we say that $\mu$ is \emph{finite}. Let us first consider contents. The following definition is a modification of the usual divisibility property \textup{(D)} from cake-cutting \cite[p.~398]{rothe}, which is suitable for contents; our modification takes into account the lack of countable additivity.

\begin{definition}\label{def-DD}
    A content $\mu$ defined on $\Ascr$ has the property \textup{(DD)} if for every $A\in\Ascr$ and $\alpha\in(0,1)$ there is an increasing sequence of sets $B^1_\alpha  \subset B^2_\alpha \subset B^3_\alpha \subset \dots$, $B^n_\alpha \in\Ascr$, such that $B^n_\alpha\subset A$ and $\sup_{n\in\nat}\mu(B^n_\alpha) = \alpha\mu(A)$.
\end{definition}
If $\mu$ is a measure on a $\sigma$-algebra $\Ascr$, then $B_\alpha := \bigcup_{n=1}^\infty B_\alpha^n$ is again in $\Ascr$, and by the continuity of measures we see that $\mu(B_\alpha) = \sup_{n\in\nat}\mu(B_\alpha^n)$, i.e.\ the properties \textup{(D)} and \textup{(DD)} are indeed equivalent.

Atoms exhibit, in some sense, the opposite behaviour to sets
enjoying the property \textup{(DD)}.
\begin{definition}\label{def-atom}
    Let $\Ascr$ be an algebra (or a $\sigma$-algebra) over the set $X$ and $\mu$ be a content (or a measure). A set $A\in\Ascr$ is an \emph{atom} if $\mu(A)>0$ and every $B\subset A$, $B\in\Ascr$, has measure $\mu(B)=\alpha\mu(A)$ with $\alpha=0$ or $\alpha=1$.\footnote{We use the convention that $0\cdot\infty = 0$.}
\end{definition}
If $A,B$ are atoms, then we have either $\mu(A\cap B)=0$ or $\mu(A\cap B)=\mu(A)=\mu(B)>0$; in the latter case, if $\mu(A\cap B)>0$, we call the atoms \emph{equivalent}. If $A$ and $B$ are non-equivalent, then $A$ and $B\setminus A$ are still non-equivalent and disjoint. Iterating this procedure, we can always assume that countably many non-equivalent atoms $(A_n)_{n\in\nat}$ are disjoint: just replace the atoms by $$A_1, A_2\setminus A_1,\dots, A_{n+1}\setminus\bigcup_{i=1}^n A_i, \dots.$$

Clearly, a finite content or measure can have at most $n\mu(X)$ non-equivalent atoms such that $\mu(A)\geq \frac 1n$, hence a finite content has at most countably many atoms.

Comparing Definition~\ref{def-DD} with Definition~\ref{def-atom} of an atom it is clear that \textup{(DD)} implies that $\mu$ has no atoms.

\begin{definition}\label{def-slice}
    Let $\mu$ be a finite content (or measure) on the algebra (or $\sigma$-algebra) $\Ascr$ over $X$. The set-function $\mu$ is \emph{sliceable}, if for any $\epsilon>0$ there are finitely many disjoint sets $B_i\in\Ascr$, $i=1,\dots, n=n(\epsilon)$, such that $0<\mu(B_i)\leq\epsilon$ and $X=B_1\cup\dots\cup B_n$ .

    A set $B\in\Ascr$ is \emph{$\mu$-sliceable}, if the set-function $A\mapsto\mu(A\cap B)$ is sliceable.
\end{definition}

\begin{lemma}\label{lemma-3}
    Let $\mu$ be a finite content on the algebra $\Ascr$ over the set $X$. If $\mu$ has no atoms, then it is sliceable.
\end{lemma}
\begin{proof}
    Fix $\epsilon>0$ and assume that $\mu(X)>\epsilon$, otherwise we are done.

    \medskip
    \emph{Step 1}:
    Let $Y\subset X$ be \emph{any} subset, and assume that there is some $B\subset Y$, $B\in\Ascr$, such that $\mu(B)>0$. Define
    $$
        \Fscr^Y :=\Fscr^Y_\epsilon := \{F\in\Ascr : F\subset Y, \: 0<\mu(F)\leq\epsilon\}.
    $$
    We claim that for the special choice $Y=B\in\Ascr$ the family $\Fscr^B$ is not empty.

    Since $X$ is not an atom, there is indeed a set $B\subset X$, $B\in\Ascr$, such that $0<\mu(B)<\mu(X)$, i.e.\ $\Fscr^B$ is well-defined. Similarly, there is some $F\subset B$, $F\in\Ascr$, with $0<\mu(F)<\mu(B)$.

    If $\mu(F)\leq\epsilon$, then $F\in\Fscr^B$, and we are done.

    If $\mu(F)>\epsilon$, we assume, to the contrary, that there is no subset $F'\subset F$, $F'\in\Ascr$, with $0<\mu(F')\leq\epsilon$. Since $F$ cannot be an atom, there is a subset $F'\subset F$ with $\epsilon<\mu(F')<\mu(F)$ and $\mu(F\setminus F')>\epsilon$. Iterating this with $F\rightsquigarrow F\setminus F'$ furnishes a sequence of disjoint sets $F_1=F',F_2,F_3,\dots$ with $\mu(F_i)>\epsilon$ for all $i\in\nat$. This is impossible since $\mu(F)<\infty$. So we can find some $F'\subset F\subset B$ with $0<\mu(F')\leq\epsilon$, i.e.\ $\Fscr^B$ is not empty.

    \medskip
    \emph{Step 2:}
    Define a(n obviously monotone) set-function $c(Y) := \sup_{C\in\Fscr^Y}\mu(C)$ for any $Y\subset X$; as usual $\sup\emptyset = 0$. Since $\Fscr^X$ is not empty, we can pick some $B_1\in\Fscr^X$ such that $\frac 12 c(X) < \mu(B_1) \leq\epsilon$.

    If $\mu(X\setminus B_1)\leq\epsilon$, we set $B_2 := X\setminus B_1$, otherwise we can pick some $B_2\in\Fscr^{X\setminus B_1}$ such that $\frac 12 c(X\setminus B_1) < \mu(B_2) \leq\epsilon$.

    In general, if $\mu(X\setminus (B_1\cup\dots\cup B_n))\leq\epsilon$, we set $B_{n+1}=X\setminus (B_1\cup\dots\cup B_n)$, otherwise we pick
    \begin{equation}\label{e-cond}\begin{gathered}
        B_{n+1}\in\Fscr^{X\setminus (B_1\cup\dots\cup B_n)}
        \quad\text{such that}\\
        \frac 12 c(X\setminus (B_1\cup \dots \cup B_n)) \leq \mu(B_{n+1})\leq \epsilon.
    \end{gathered}
    \end{equation}

    We are done, if our procedure stops after finitely many steps, otherwise we get a sequence of disjoint sets $B_1, B_2, \dots$ satisfying \eqref{e-cond}.
    Define $B_\infty := X\setminus\bigcup_n B_n$. This set need not be in $\Ascr$, but we still have, because of \eqref{e-cond},
    $$
        c(B_\infty)
        \leq c(X\setminus (B_1\cup\dots\cup B_m))
        \leq 2\mu(B_{n+1})
        \xrightarrow[n\to\infty]{}0
    $$
    since the series
    $$
        \sum_{n=1}^\infty \mu(B_n)
        = \sup_N\sum_{n=1}^N \mu(B_n)
        = \sup_N\mu(\bigcup_{n=1}^N B_n)
        \leq \mu(X)%<\infty
    $$
    converges.
    In particular, $\lim_{n\to\infty}\mu(X\setminus \bigcup_{i=1}^nB_i)=0$.

    Using again the convergence of the series $\sum_n \mu(B_n)$, we find some $N=N(\epsilon)$ such that $\sum_{N(\epsilon)< n\leq\infty}\mu(B_n)\leq\epsilon$, hence $B_1,B_2,\dots, B_N$ and $X\setminus\bigcup_{n=1}^N B_n$ is the slicing of $X$.
\end{proof}

\begin{lemma}\label{lemma-4}
    Let $\mu$ be a finite content without atoms on an algebra $\Ascr$ over the set $X$. Then
    it is sliceable and enjoys the property \textup{(DD)}.
\end{lemma}
\begin{proof}
    Sliceability follows from Lemma~\ref{lemma-3}. Let us show that this implies condition \textup{(DD)}. Let $B\in\Ascr$ such that $\mu(B)>0$. Since the content $\mu_B(A):=\mu(A\cap B)/\mu(B)$ inherits the non-atomic property from $\mu$, it is clearly enough to show that for every $\alpha\in(0,1)$ there is an increasing sequence
    $$
        B^1_\alpha\subset B^{2}_\alpha\subset B^{3}_\alpha\subset\dots,
        \quad B^n_\alpha\in\Ascr \::\: \sup_{n\in\nat}\mu(B^n_\alpha) = \alpha,
    $$
    which is property \textup{(DD)} only relative to the full space $X$.

    Since $X$ is $\mu$-sliceable, there are mutually disjoint sets $C_1^{n}, \dots, C_N^{n}\in\Ascr$ where $N=N(n)$, $X=\bigcup_{i=1}^N C_i^{n}$ and $\mu(C_i^{n})<\frac 1n$.

    Let $k = \entier(1/\alpha)+1$ where $\entier(x)$ denotes the integer part of $x\in\real$.

    Set $B_k := C_1^{k}\cup\dots\cup C_{M(k)}^{k}$ where $M(k)\in \{1,\dots, N(k)\}$ is the unique number such that
    $$
        \sum_{i=1}^{M(k)}\mu(C_i^k)\leq\alpha <\sum_{i=1}^{M(k)+1}\mu(C_i^k)\leq \sum_{i=1}^{M(k)}\mu(C_i^k)+\frac 1k.
    $$

    Cy construction, $\mu(B_k)=\sum_{i=1}^{M(k)+1}\mu(C_i^k) > \alpha-\frac 1k$. Thus, we can iterate this procedure, considering $X\setminus B_k$ and constructing a set $D_{k+1}\subset X\setminus B_k$ which satisfies
    $$
        (\alpha-\mu(B_k)) \geq \mu(D_{k+1}) > (\alpha-\mu(B_k))-\frac 1{k+1}.
    $$
    For $B_{k+1}:= B_k\cup D_{k+1}$ we get $\alpha\geq \mu(B_{k+1})>\alpha-\frac 1{k+1}$.

    The sequence $B_{k+i}$, $i\in\nat$, satisfies $\mu(B_{k+1})\leq\alpha$, and so $\mu(B_{k+i})\uparrow\alpha$, i.e.\ $B_\alpha^n = B_{k+n}$ is the sequence of sets we need.
\end{proof}

If we collect all of the above results, we finally arrive at
\begin{theorem}\label{theo-conti}
    Let $\mu$ be a finite content on an algebra $\Ascr$ over $X$. The condition \textup{(DD)} is equivalent to `$\mu$ is sliceable' and `$\mu$ has no atoms'.
\end{theorem}

Any measure $\mu$ on a $\sigma$-algebra is trivially a content on an algebra, i.e.\ Lemmas~\ref{lemma-3} and \ref{lemma-4} remain valid for measures. Since the properties \textup{(D)} and \textup{(DD)} are equivalent for measures defined on a $\sigma$-algebra, we immediately get
\begin{corollary}\label{cor-conti}
    Let $\mu$ be a finite measure on a $\sigma$-algebra $\Ascr$ over $X$. The condition \textup{(D)} is equivalent to `$\mu$ is sliceable' and `$\mu$ has no atoms'.
\end{corollary}
If $A_1, A_2, \dots$ are an enumeration of the non-equivalent atoms of the
measure $\mu$, then $A_\infty := X\setminus \bigcup_{n\in\nat}A_n\in\Ascr$ and
we can re-state Corollary~\ref{cor-conti} in the form of a decomposition theorem.
\begin{corollary}\label{cor-deco}
    Let $\mu$ be a finite measure on a $\sigma$-algebra $\Ascr$ over $X$. Then $X$ can be written as a disjoint union of a $\mu$-sliceable set $S$ and at most countably many atoms $A_1, A_2, \dots$.
\end{corollary}

\section{Conclusions}\label{conc}
The present paper has led to clear answers on the three questions posed in the introduction:
\begin{enumerate}
\item%[1.]
    The system of Borel sets is the correct system of subsets of $[0,1]$.

\item%[2.]
    Valuations should be considered as measures, not as contents. For one-dimensional cakes the divisibility assumption \textup{(D)}
    automatically entails $\sigma$-additivity.

\item%[3.]
    The divisibility assumption \textup{(D)} is equivalent to the assumption of non-existence of atoms. The assumption that there is a density function is a strictly stronger requirement.
\end{enumerate}

Depending on the intended audience we propose the following two versions to introduce the problem of cake-cutting.
\begin{quote}
\textbf{Version 1}: A cake is a one-dimensional heterogeneous good, represented by the unit interval $[0,1]$. Each of $n$ players has a personal value
function over the cake, characterised by a probability density function with respect to Lebesgue measure. This implies that players' preferences are $\sigma$-additive (i.e.\ countably additive) and non-atomic.
\end{quote}

For a more mathematically-minded audience, the following general version might be used.
\begin{quote}
\textbf{Version 2}:
We consider a cake which is represented by the interval $[0,1]$. A \emph{piece of cake} is a Borel subset of $[0,1]$. We will make the standard assumptions in cake-cutting. Each agent in the set of agents $N=\{1, \ldots n\}$ has his own valuation over Borel subsets of the interval $[0,1]$. The valuations are Borel measures, i.e.\ they are (i) \emph{defined on all Borel sets}; (ii) \emph{non-negative}: $V_i(X) \ge 0$ for all Borel subsets $X$ of $[0,1]$; (iii) \emph{$\sigma$-additive}: $V_i\left(\bigcup_{n=1}^\infty X_n\right) = \sum_{n=1}^\infty V_i(X_n)$ for all sequences of disjoint Borel sets $X_1, X_2, \ldots$; (iv) \emph{divisible} i.e.\ for every Borel set $X$ in $[0,1]$ and $0 \le \alpha \le 1$, there exists a Borel set $X' \subset X$ with $V_i(X') = \alpha V_i(X)$.
\end{quote}
Here it could be noted that (iv) is equivalent to having no atoms, i.e.\ the (cumulative) distribution function is continuous. The latter property, if combined with finite additivity, already entails $\sigma$-additivity (iii).

For readers not familiar with measure theory it could be noted that the system of all Borel sets is stable under countable repetition of all set-theoretic operations (union, intersection,  formation of complements) and that it contains all intervals, all open and all closed sets. There is, however, no constructive way to build a general Borel set starting from, say, the intervals.

\section*{Acknowledgement}
We are grateful to Steven J. Brams (NYU, New York), J\"{o}rg Rothe (U D\"{u}sseldorf) and Walter Stromquist (Swarthmore College, Swarthmore, PA) for valuable discussions and suggestions on the topic of this paper; Bj\"{o}rn B\"{o}ttcher (TU Dresden) helped us to produce Figures~1 and 2, and proofread the whole paper.

%\enlargethispage\bigskipamount

\end{document}